\documentclass{siamart220329}

\usepackage{graphicx, amsmath, amssymb, a4, wasysym}  
\usepackage[utf8]{inputenc}
\usepackage[sort]{cite}
\usepackage[toc]{appendix}

\usepackage{hyperref}

\graphicspath{{figures/}}

\newcommand{\R}{\mathbb{R}}

\renewcommand{\phi}{\varphi}

\newcommand{\calF}{\mathcal{F}}

\newcommand{\calR}{\mathcal{R}}
\newcommand{\calS}{\mathcal{S}}
\newcommand{\calV}{\mathcal{V}}

\newcommand{\abs}[1]{\left|#1\right|}

\DeclareMathOperator{\pder}{\partial}

\DeclareMathOperator{\tr}{tr}

\newcommand{\lb}{\left(}
\newcommand{\lsb}{\left[}
\newcommand{\lcb}{\left\{}

\newcommand{\rb}{\right)}
\newcommand{\rsb}{\right]}
\newcommand{\rcb}{\right\}}

\newcommand{\lpm}{\begin{pmatrix}}
\newcommand{\rpm}{\end{pmatrix}}

\newsiamremark{example}{Example}
\newsiamremark{remark}{Remark}

\title{SIR--Model for Households}
\author{Philipp Dönges\thanks{Max Planck Institute for dynamics and self--organization, Göttingen, Germany	(\email{philipp.doenges@ds.mpg.de}; \email{viola.priesemann@ds.mpg.de}).}
\and Thomas Götz\thanks{Mathematical Institute, University Koblenz, Germany 
	(\email{goetz@uni-koblenz.de}; \email{nkruchinina@uni-koblenz.de}; \email{moritzschaefer@uni-koblenz.de}).}
\and Nataliia Kruchinina\footnotemark[2]
\and Tyll Krüger\thanks{Polytechnic University Wroclaw, Poland
	(\email{tyll.krueger@pwr.edu.pl}).}
\and Karol Niedzielewski\thanks{Interdisciplinary Centre for Mathematical and Computational Modelling (ICM), University of Warsaw, Poland
		(\email{k.niedzielewski@icm.edu.pl}).}
\and Viola Priesemann\footnotemark[1]
\and Moritz Schäfer\footnotemark[2]}

\headers{SIR--Model for households}{P.~Dönges, T.~Götz, T.~Krüger, e.a.}

\begin{document}

\maketitle

\begin{abstract}
Households play an important role in disease dynamics. Many infections happening there due to the close contact, while mitigation measures mainly target the transmission between households. Therefore, one can see households as boosting the transmission depending on household size. To study the effect of household size and size distribution, we differentiated the within and between household reproduction rate. There are basically no preventive measures, and thus the close contacts can boost the spread.
We explicitly incorporated that typically only a fraction of all household members are infected. 
Thus, viewing the infection of a household of a given size as a splitting process generating a new, small fully infected sub--household and a remaining still susceptible sub--household we derive a compartmental ODE--model for the dynamics of the sub--households. In this setting, the basic reproduction number as well as prevalence and the peak of an infection wave in a population with given households size distribution can be computed analytically. We compare numerical simulation results of this novel household--ODE model with results from an agent--based model using data for realistic household size distributions of different countries. We find good agreement of both models showing the catalytic effect of large households on the overall disease dynamics.
\end{abstract}



\section{Introduction}

The spread of an infectious diseases strongly depends on the interaction of the considered individuals. Traditional SIR--type models assume a homogeneous mixing of the population and typically neglect the increased transmission within closed subcommunities like households or school--classes. However, literature indicates, that household transmission plays an important role~\cite{fraser2011influenza, del2022secondary, jorgensen2022secondary}.

In case of the COVID--pandemic, several studies have quantified the secondary attack rate within households, i.e. the probability household-members get infected, given that one household member is infected \cite{li2020characteristics,madewell2020household,house2022inferring,del2022secondary,jorgensen2022secondary,lyngse2022household}. The secondary attack rate depends on the virus variant, immunity and vaccination status, cultural differences and mitigation measures within a household (like e.g. early quarantine). In December 2021, when both the Delta and Omicron variants were spreading, it ranged between about 19 \% (Delta variant in Norway) \cite{jorgensen2022secondary} to 39 \% (Omicron in Spain) \cite{del2022secondary}. Hence, interestingly, the secondary attack rate within a household is far from 100 \% despite the close contacts. 

There are several models reported that try to include the contribution of in--household transmission to the overall disease dynamics. The Reed--Frost model describing in--household infections as a Bernoulli--process has be used by~\cite{glass2011incorporating, fraser2011influenza}. An average model assuming households always get completely infected and ignoring the temporal dynamics has been proposed in~\cite{becker1995effect}. Their findings for the effective reproduction number agree with our results, see~\eqref{E:tildeR}. Extensions of differential equation based SIR--models in case of a uniform household distribution have be proposed by~\cite{house2008deterministic,huber2020minimal}. However, in real populations, the household distribution is far from uniform, see~\cite{unstats2015demographic}. Ball and co--authors provided in ~\cite{ball2015seven} an overview of challenges posed by integrating household effects into epidemiological models, in particular in the context of compartmental differential equations models. 

In this paper we develop an extended ODE SIR--model for disease transmission within and between households of different sizes. We will treat the infection of a given household as a splitting process generating two new sub--households of smaller size - representing the susceptible ($S$) or fully infected ($I$) members. Each of the susceptible sub--households can get infected and split later in time, while infected sub--households recover with a certain rate and are then immune (recovered or removed, $R$). The dynamics of the susceptible, infected and recovered sub--households is modeled in Section~\ref{S:Model}. In Section~\ref{S:R0}, we compute the basic reproduction number for our household model combining the attack rate \textit{inside} a single households and the transmission rate \textit{between} individual households. 
The prevalence, as the limit of the recovered part of the population can be computed analytically --- at least in case of small maximal household size, see Section~\ref{S:Prevalence}. The peak of a single epidemic wave in a population with given household distribution is considered in Section~\ref{S:Peak}. Again, for small maximal household size, we will be able to compute analytically the maximal number of infected. Numerical simulations based on realistic household size distributions for different countries and a comparison with an agent--based model demonstrate the applicability of the presented model. As outlook with focus on non--pharmaceutical interventions we consider an extension of our model quarantining infected sub--households based on a certain detection rate for a single infected. In an appendix, we present an alternative derivation of analytical results for the basic reproduction number and the prevalence at the end of the epidemic based on random graph theory. These results confirm the findings based on the presented compartmental ODE model.

\section{Household Model}
\label{S:Model}

To model the infection process inside household, we view an infection event in a fully susceptible household of size $k$ as a splitting process generating a infected sub--household of size $j$, where $1\le j\le k$ and a remaining susceptible sub--household of size $k-j$. Let $S_j, I_j$ and $R_j$ denote the number of sub--households of size $j$, where $1\le j \le K$ and $K$ denotes the maximal household size. Then $H_j=S_j+I_j+R_j$ equals to the total number of current sub--households of size $j$. Furthermore, we introduce $H=\sum_{j=1}^K H_j$ as the total number of households and $h_j=H_j/H$. The total population is given by $N=\sum_{j=1}^K j H_j$. For further reference we also introduce the first two moments of the household size distribution
\begin{align*}
	\mu_1 := \sum_{j=1}^K j h_j \quad \text{and} \quad
	\mu_2 := \sum_{j=1}^K j^2 h_j\;.
\end{align*}
If an initial infection is brought into a susceptible sub--household of size $j$, secondary infections will occur inside the household. We assume that each of the remaining $j-1$ household members can get infected with equal probability $a$, called the \emph{in--household attack rate}. Hence we expect in total
\begin{equation*}
	E_j = a(j-1)+1
\end{equation*}
infections (including the primary one) inside a household of size $j$. Existing field studies, see~\cite{madewell2020household, house2022inferring} indicate an in--household attack rate in the range of $16\%$--$30\%$ depending on the overall epidemiological situation, household size and vaccinations. For the sake of tractability and simplicity, our model assumes a constant attack rate $a$ independent of the household size.

Let $b_{j,k}$ denote the probability, that a primary infection in a household of size $j$ generates in total $k$ infections inside this household, where $1\le k\le j$. The secondary infections give rise to a splitting of the initial household of size $j$ into a new, fully infected sub--household of size $k$ and another still susceptible sub--household of size $j-k$.

An infected household of size $k$ recovers with a rate $\gamma_k$ and contributes to the overall force of infection between different households by a so--called out--household infection rate $\beta_k$. The term ''out--household'' refers to infection events occurring between different households and hence \emph{outside} a given single household. We assume that the out--household reproduction number is independent of the household size, i.e. 
\begin{equation}
\label{E:Rstar_const}
	\calR^\ast = \frac{\beta_k}{\gamma_k} = \text{constant independent of $k$}\;.
\end{equation}

Now, the dynamical system governing the dynamics of the susceptible, infected and recovered households of size $k$ reads as
\begin{subequations}
\label{E:SIR-HH}
\begin{align}
	S_k' &= Y \lsb - k S_k + \sum_{j=k+1}^K\!\! j S_j \cdot b_{j,j-k} \rsb\;,  \label{E:SIR-HH-S}\\
	I_k' &= - \gamma_k I_k + Y \sum_{j=k}^K j S_j \cdot  b_{j,k}\;, \\
	R_k' &= \gamma_k I_k\;,
	\intertext{where}
	Y &:= \frac{1}{N} \sum_{k=1}^K \beta_k \cdot k I_k \notag
\end{align}
\end{subequations}
denotes the total force of infection.

For the recovery rate of an infected household of size $k$ we can consider the two extremal cases and an intermediate case.
\begin{enumerate}
\item \emph{Simultaneous} infections: All members of the infected household get infected at the same time and recover at the same time, hence the recovery rate $\gamma_k=\gamma_1$ is independent of the household size. Assumption~\eqref{E:Rstar_const} leads to a constant out--household infection rate $\beta_k=\beta$.
\item \emph{Sequential} infections: All members of the household get infected one after another and the total recovery time for the entire household equals to $k$ times the individual recovery time. Hence the recovery rate $\gamma_k=\gamma_1/k$ and by~\eqref{E:Rstar_const} we get $\beta_k=\beta_1/k$.
\item \emph{Parallel} infections: The recovery times $T_i$, $i=1,\dots, k$ for each of the $k$ infected individuals are modeled as independent exponentially distributed random variables. Hence the entire household is fully recovered at time $\max(T_1,\dots, T_k)$. The recovery rate $\gamma_k$ equals to the inverse of the expected recovery time, i.e.~
\begin{equation*}
	\gamma_k = \frac{1}{E[\max(T_1,\dots, T_k)]} = \frac{1}{E[T_1]\cdot \theta_k}
		= \frac{\gamma_1}{\theta_k}\;,
\end{equation*}
where $\theta_k=1+\tfrac12+\dots + \tfrac1k \sim \log k + g$ denotes the $k$--th harmonic number and $g\approx 0.5772\dots$ denotes the \emph{Euler--Mascheroni} constant.
\end{enumerate}
All these cases can be subsumed to
\begin{equation*}
	\gamma_k = \frac{\gamma_1}{\eta_k}, \quad \beta_k = \frac{\beta_1}{\eta_k}\;,
\end{equation*}
where $\eta_k$ models the details of the temporal dynamics inside an infected household.

In Figure~\ref{F:I_Gamma} we compare these three cases in the scenario of a population with maximal household size $K=6$. Each household size represents $1/6$ of the entire population. Initially, $1\permil$ of the population is infected. The out--household reproduction number is assumed to be $\calR^\ast=1.33$, in--household attack rate equals $a=0.2$ and the recovery rate $\gamma_1$ for an infected individual equals $0.1$. Shown are the incidences, i.e.~the daily new infections over time. The three cases differ in the timing and the height of the peak of the infection; for the simultaneous infections ($\gamma_k$ constant) the disease spreads fastest and for the sequential infections ($\gamma_k=\gamma_1/k$) the spread is significantly delayed. The cases of parallel infections, i.e.~$\gamma_k\simeq \gamma_1/(\log k+g)$, lies between theses two extremes.  
\begin{figure}[htb]
\centering\includegraphics[width=0.67\textwidth]{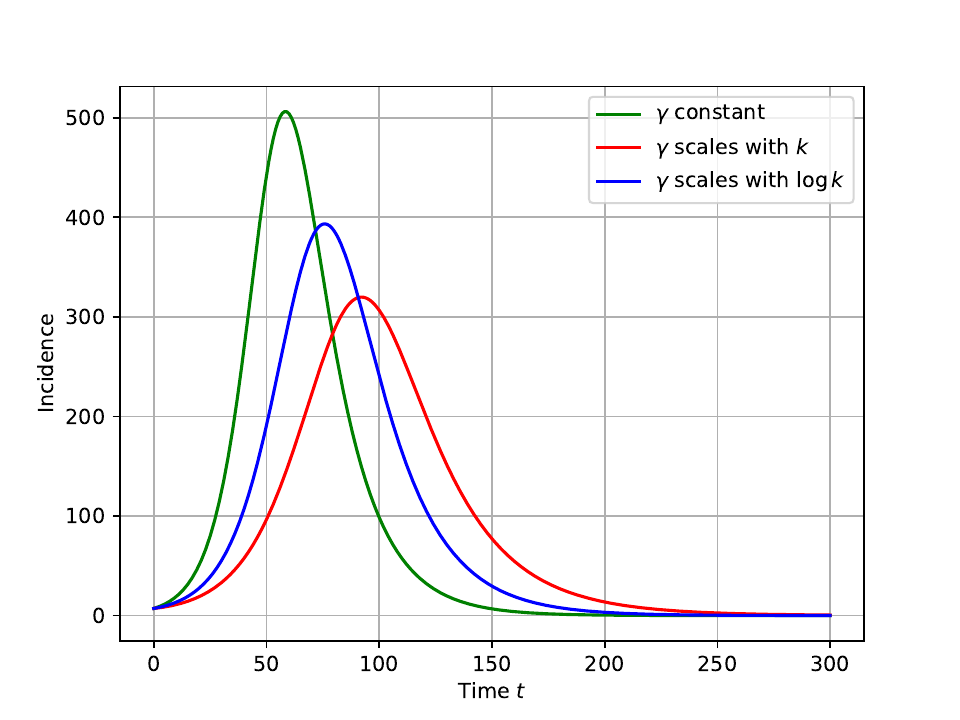}  
\caption{\label{F:I_Gamma} Simulation of one epidemic wave in case of the three different scalings of the recovery rates for  households of size $k$. Shown is the incidence, i.e.~daily new infections over time. The green curve corresponds equal recovery rates for all households, i.e.~$\gamma_k=\gamma_1$. The red curve corresponds to $\gamma_k=\gamma_1/k$ (sequential infections). The blue curve corresponds to recovery rates obtained from the maximum of exponential distributions, i.e.~$\gamma_k \simeq \gamma_1/(\log k)$.}
\end{figure}

\begin{example}
Modeling the infections inside the household by a Bernoulli--process with in--household attack rate (infection probability) $a\in [0,1]$, the total number of infected persons inside the household follows a binomial distribution
\begin{equation}
\label{E:bjk}
	 b_{j,k} := \binom{j-1}{k-1} a^{k-1} (1-a)^{j-k}\;.
\end{equation}
For the binomial in--household infection~\eqref{E:bjk} it holds, that $E_j=a (j-1)+1$.
\end{example}

\begin{theorem}\label{H:Nconst} In model~\eqref{E:SIR-HH} the total population $N=\sum_{k=1}^K k\cdot(S_k+I_k+R_k)$ is conserved.
\end{theorem}
\begin{proof}
We have 
\begin{equation*}
	N' = Y \lsb -\sum_{k=1}^K k^2 S_k
		+ \sum_{k=1}^K k \sum_{j=k+1}^K j S_j b_{j,j-k} 
		+ \sum_{k=1}^K k \sum_{j=k}^K j S_j b_{j,k}\rsb
\end{equation*}
We consider the last two summands separately and reverse the order of summation. Hence
\begin{align*}
	\sum_{k=1}^K k \sum_{j=k}^K j S_j b_{j,k} 
	&= \sum_{j=1}^K j S_j \sum_{k=1}^j k b_{j,k} 
		= \sum_{j=1}^K j E_j S_j
\intertext{and analogously}
	\sum_{k=1}^K k \sum_{j=k+1}^K j S_j b_{j,j-k} &=
		\sum_{j=2}^K j S_j \sum_{k=1}^{j-1} k b_{j,j-k} \\
		&= \sum_{j=2}^K j S_j \sum_{l=1}^{j-1} (j-l) b_{j,l} 
		= \sum_{j=2}^K j S_j \sum_{l=1}^{j} (j-l) b_{j,l} \\
		&= \sum_{j=2}^K j S_j \lb j \sum_{l=1}^j b_{j,l} - \sum_{l=1}^j l b_{j,l} \rb 
		= \sum_{j=2}^K j^2 S_j - j E_j S_j \\
		&= \sum_{j=1}^K j^2 S_j - j E_j S_j - \lb S_1 - E_1 S_1 \rb\;.
\end{align*}		
Due to $E_1=1$, the last term vanishes. Summing all the contributions, we finally get
\begin{equation*}
	N' = Y\lsb \sum_{j=1}^K -(j^2 S_j) +  j E_j S_j + (j^2 S_j - j E_j S_j)\rsb = 0\;.
\end{equation*}
\end{proof}

\section{Basic Reproduction Number}
\label{S:R0}

To compute the basic reproduction number for the household--model~\eqref{E:SIR-HH}, we follow the next generation matrix approach by Watmough and van den Driessche, see~\cite{van2002reproduction}. We split the state variable $x=(S_1, \dots, R_K)\in \R^{3K}$ into the infected compartment $\xi=(I_1,\dots I_K)\in \R^K$ and the remaining components $\chi\in \R^{2K}$. The infected compartment satisfies the differential equation
\begin{equation*}
	\xi_k' = \calF_k(\xi,\chi) - \calV_k(\xi,\chi)
\end{equation*}
where $\calF_k = Y \sum_{j=k}^K j S_j \,b_{j,k}$, $\calV_k=\gamma_k \xi_k$ and $Y=\tfrac{1}{N} \sum_{j=1}^K \beta_j j \xi_j$. Now, the Jacobians at the disease free equilibrium $(0,\chi^\ast)$ are given by
\begin{align*}
	F_{kj} &= \frac{\pder \calF_k}{\pder\xi_j}(0,\chi^\ast) 
		= \frac{1}{N}j \beta_j \sum_{m=k}^K m H_m \, b_{m,k} \\
	V_{kk} &= \frac{\pder \calV_k}{\pder \xi_k}(0,\chi^\ast)  = \gamma_k 
		\quad \text{and $V_{kj}=0$ for $k\neq j$}
\end{align*}

The next generation matrix $G\in \R^{K\times K}$ is given by 
\begin{align*}
	G &= FV^{-1} = \frac{1}{N} 
		\lb \frac{j \beta_j}{\gamma_j}\sum_{m=k}^K m H_m\, b_{m,k} \rb_{k,j}
\intertext{Using the relation~\eqref{E:Rstar_const} we obtain}
	G &= \calR^\ast
		\lb j\sum_{m=k}^K \frac{m H_m}{N}\, b_{m,k} \rb_{k,j}
		= \calR^\ast 
		\lpm \sum_{m=1}^K \tfrac{m H_m}{N}\, b_{m,1} \\[.5ex] 
			\sum_{m=2}^K \tfrac{m H_m}{N} \, b_{m,2} \\\vdots \\
			\tfrac{K H_K}{N} \,  b_{KK} \rpm 
	\cdot \lpm 1, 2, \dots, K\rpm
\end{align*}
where the last term represents the next generation matrix as a dyadic product $G=\calR^\ast \cdot dc^T$.

The basic reproduction number $\calR=\rho(G)$ is defined as the spectral radius $\rho(G)$ of the next generation matrix $G\in \R^{K\times K}$. As a dyadic product, $G$ has rank one and hence there is only one non--zero eigenvalue.

\begin{lemma}
Let $c,d\in \R^n$ be two vectors with $c^T d\neq 0$. The non--zero eigenvalue of the dyadic product $A=dc^T\in \R^{n\times n}$ is given by $c^T d=\tr(A)$.
\end{lemma}
\begin{proof}
Let $v$ be an eigenvector of $A$ to the eigenvalue $\lambda$. Then $\lambda v=Av=(dc^T)v = (c^Tv) d$, where $c^Tv\in \R$. If $\lambda\neq 0$, then $v=\frac{c^T v}{\lambda} d$. Set $\mu=\frac{c^T v}{\lambda}\in \R$. Now, $Av=A(\mu d) = \mu (dc^T)d = (c^T d) \mu d  = (c^T d) v$; hence $\lambda=c^T d$. 
\end{proof}

As an immediate consequence we obtain
\begin{theorem} The basic reproduction number of system~\eqref{E:SIR-HH} is given by
\begin{equation}
\label{E:R0}
	\calR = \calR^\ast \sum_{i=1}^K i \sum_{m=i}^K \frac{m H_m}{N} b_{m,i}
	 = \calR^\ast \sum_{m=1}^K \frac{m H_m}{N} E_m\;.
\end{equation}
\end{theorem}

\begin{corollary}
In the particular situation, when the expected number of infections inside a household of size $m$ is given by $E_m=a(m-1)+1$, the basic reproduction number equals
\begin{equation}
\label{E:tildeR}
	\calR = \calR^\ast\lsb 1+a\lb \frac{\mu_2}{\mu_1}-1\rb\rsb\;.
\end{equation}
\end{corollary}

Assuming, that the infection of any member of a household results in the infection of the entire household, i.e.~in--household attack rate $a=1$, our result $\calR=\frac{\mu_2}{\mu_1}\calR^\ast$ agrees with the result obtained by Becker and Dietz in~\cite[Sect. 3.2]{becker1995effect}. In the appendix~\ref{S:App1} we show, that a random graph model for the household infection process leads to the same result for the reproduction number as epidemic threshold, see~\eqref{E:normT}.

\section{Computing the prevalence}
\label{S:Prevalence}

Let $z = \frac{1}{N}\sum_{k=1}^K k R_k$ denote the fraction of recovered individuals. Then $z$ satisfies the ODE
\begin{align*}
	z' &= \frac{Y}{\calR^\ast}\;.
\end{align*}

In the sequel we will derive an implicit equation for the prevalence $\lim_{t\to\infty} z(t)$ in two special cases:
\begin{enumerate}
\item for maximal household size $K=3$. The procedure used here allows for immediate generalization but the resulting expression gets lengthy and provide only minor insight into the result.
\item for in--household attack rate $a=1$ and arbitrary household sizes. In this setting the equations~\eqref{E:SIR-HH-S} for the susceptible households decouple and allow and complete computation of the equation for the prevalence.
\end{enumerate}

We will start with the second case. Let us consider $a=1$, i.e.~inside households infections are for sure and $E_k=k$. Then the ODE for the susceptible households reads as $S_k'=-Y k S_k$ and we can insert the  recovered $z$ 
\begin{align*}
	S_k' &= -k \calR^\ast S_k z\;.
\intertext{After integration with respect to $t$ from $0$ to $\infty$, we get}
	\ln \frac{S_k(\infty)}{S_k(0)} &= k\calR^\ast\lb z(0)-z(\infty) \rb\;.
\end{align*}
Assuming initially no recovered individuals, i.e.~$z(0)=0$ and considering the total population at the the end of time, i.e.~$N=Nz(\infty) + \sum_k k S_k(\infty)$, we arrive at the system
\begin{align}
	N &= Nz(\infty) + \sum_{k=1}^K k S_k(0) e^{-k \calR^\ast z(\infty)}\;. \notag
\intertext{Scaling with $N$, i.e.~introducing $s_{k,0} = S_k(0)/N$, we get}
	1 &= z+ \sum_{k=1}^K k\, s_{k,0}\, e^{-k \calR^\ast z}\;. \label{E:preval}
\intertext{In the limit $\calR^\ast z \ll 1$ we obtain the approximation}
	z &\sim \frac{1-\sum_k k s_{k,0}}{1-\calR^\ast \sum_k k^2 s_{k,0}}\;. \notag
\end{align}
\begin{figure}[htb]
\centering\includegraphics[width=0.67\textwidth]{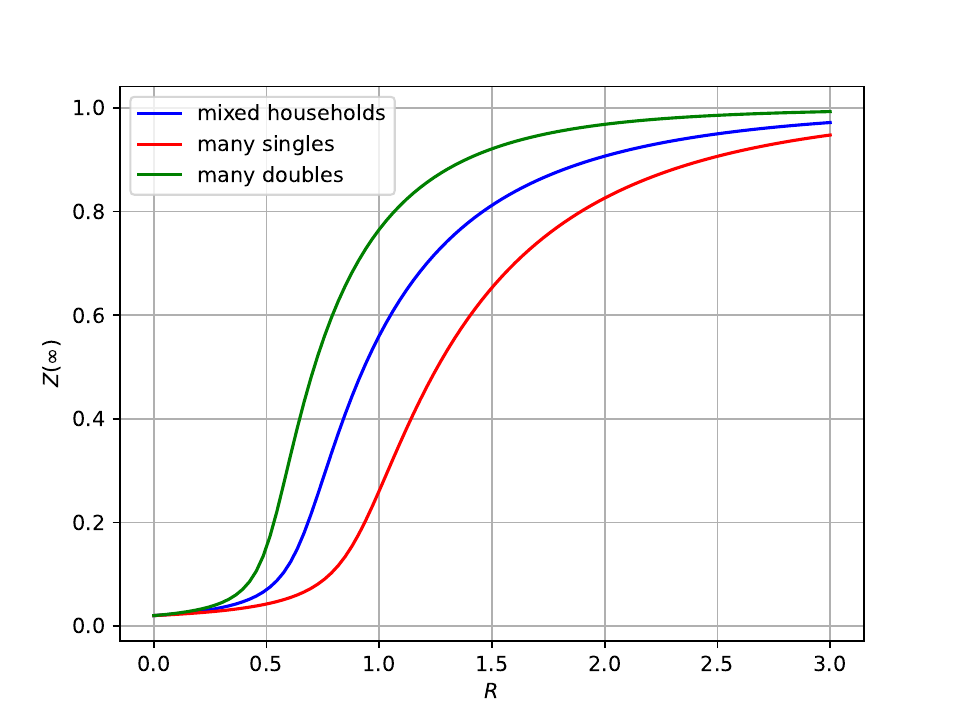}  
\caption{\label{F:preval1} Prevalence in case of maximal household size $K=2$ vs. reproduction number. If double households dominate (green), the prevalence is larger than in the case of mostly single households (red), since households speed up the infection dynamics.}
\end{figure}
Figure~\ref{F:preval1} shows the numerical solution of the prevalence equation~\eqref{E:preval} in case of $K=2$ for reproduction numbers $\calR^\ast \in [0, 3]$ and initially $2\%$ of the entire population being infected. The three different graphs correspond to different initial values for the susceptible households: both single and double households contain $49\%$ of the population (blue), $89\%$ of susceptibles live in single households and only $9\%$ in double households (red) or just $9\%$ in single households and $89\%$ in double households (green).

In case of arbitrary in--household attack rates $a\in [0,1]$, we consider the situation for $K=3$. Setting $z=Z/N$, the relevant equations read as 
\begin{align*}
	S_3' &= -3 \calR^\ast z'\,  S_3 \\
	S_2' &= \calR^\ast z' \lb -2 S_2 + 3 b_{3,1} S_3\rb\\
	S_1' &= \calR^\ast z' \lb -S_1 +2b_{2,1}S_2 + 3b_{3,2}S_3\rb\;.
\end{align*}
Solving the equations successively starting with $S_3(t) = S_3(0) e^{-3\calR^\ast z(t)}$ and using variation of constants, we arrive at
\begin{align*}
	S_2(t) &= S_2(0) e^{-2\calR^\ast z(t)} 
		+ 3b_{3,1} S_3(0) \lb 1-e^{-\calR^\ast z(t)}\rb e^{-2\calR^\ast z(t)} \\
	S_1(t) &= S_1(0) e^{-\calR^\ast z(t)} 
		+ 2 b_{2,1}\lb S_2(0) + 3 b_{3,1} S_3(0)\rb \lb 1-e^{-\calR^\ast z(t)}\rb e^{-\calR^\ast z(t)} \\
		& \qquad + \tfrac{3}{2} \lb b_{3,2} - 2 b_{2,1} b_{3,1}\rb S_3(0) \lb 1-e^{-2\calR^\ast z(t)}\rb e^{-\calR^\ast z(t)}\;.
\end{align*}
For the prevalence $z=\lim_{t\to\infty} z(t)$ we obtain the implicit equation
\begin{multline}
	\label{E:z}
	1 = z+ \lsb s_{1,0} + 2b_{2,1}s_{2,0} + 3(b_{2,1}b_{3,1}+\tfrac12 b_{3,2})s_{3,0}\rsb e^{-\calR^\ast z} \\
	\qquad + 2(1-b_{2,1}) \lsb s_{2,0} + 3 b_{3,1} s_{3,0}\rsb e^{-2\calR^\ast z} \\
	+ 3\lsb 1+ (b_{2,1}-2) b_{3,1}- \tfrac12 b_{3,2}\rsb s_{3,0} e^{-3\calR^\ast z}
\end{multline}
or in short
\begin{equation*}
	1 = z + c_1 e^{-\calR^\ast z} + c_2 e^{-2\calR^\ast z} + c_3 e^{-3\calR^\ast z}\;, 
\end{equation*}
where the coefficients $c_1, c_2$ and $c_3$ are the above, rather lengthy expressions involving the initial conditions and the in--household infection probabilities $b_{j,k}$. In Appendix~\ref{S:App2} we will derive expressions for the prevalence using tools from random graph models. This alternative setting allows to compute the same results almost explicitly even for lager household sizes $k\ge 3$. In case of $K=2$, i.e.~setting $s_{3,0}=0$, we get
\begin{equation}
	\label{E:preval2}
	1 = z + \lb s_{1,0}+2b_{2,1} s_{2,0}\rb e^{-\calR^\ast z} 
		+ 2 \lb 1-b_{2,1}\rb s_{2,0} e^{-2\calR^\ast z}\;.
\end{equation}
In case of arbitrary household size $K>3$, the resulting equation for the prevalence will have the same structure. 

For the binomial infection distribution with $a=1$ this reduces to
\begin{align*}
	1 &= z+ s_{1,0}e^{-\calR^\ast z}+ 2s_{2,0}e^{-2\calR^\ast z}+ 3s_{3,0}e^{-3\calR^\ast z}
\intertext{and in case of $a=0$ we arrive at}
	1 &= z+ \lb s_{1,0}+ 2 s_{2,0} + 3 s_{3,0}\rb e^{-\calR^\ast z}\;.
\end{align*}

In case of small initial infections, i.e.~$s_{1,0}+2s_{2,0}+3s_{3,0}=c_1+c_2+c_3\approx 1$, Eqn.~\eqref{E:z} allows for the trivial, disease free solution $z=0$. However, above the threshold $\calR_c = \frac{1}{c_1+2c_2+3c_3}$ the non--trivial endemic solution shows up. Expanding the exponential for $\calR^\ast z\ll 1$, we arrive at
\begin{gather*}
	1 \approx z + \lb c_1+c_2+c_3 \rb - \calR^\ast z \lb c_1+2c_2+3c_3\rb + \frac{{\calR^\ast}^2 z^2}{2} \lb c_1+ 4c_2 + 9 c_3\rb 
	\intertext{and hence}
	z \lb \calR^\ast (c_1+2c_2+3c_3)-1\rb \approx \frac{{\calR^\ast}^2 z^2}{2} \lb c_1+ 4c_2 + 9 c_3\rb\;.
\end{gather*}	
Besides the trivial solution $z=0$, this approximation has the second solution
\begin{equation}
\label{E:preval3_appx}
	z\simeq \frac{2\lb \calR^\ast (c_1+2c_2+3c_3) -1 \rb}{{\calR^\ast}^2 \lb c_1+4c_2+9c_3\rb}\;.
\end{equation}
If $\calR^\ast>\calR_c = (c_1+2c_2+3c_3)^{-1}$, this second root is positive.

The following Figure~\ref{F:preval}(left) shows the numerical solution of the prevalence equation~\eqref{E:preval2} in case of $K=2$ for reproduction numbers $\calR^\ast \in [0, 3]$ and initially $2\%$ of the entire population being infected while both single and double households contain $49\%$ of the susceptible population. The different curves show the results depending on the in--household attack rate $a$. The graph on the right depict the case $K=3$ and $s_{1,0}=s_{2,0}=s_{3,0}=\frac{1}{6}$ for three different in--household attack rates $a$. The approximation~\eqref{E:preval3_appx} is shown by the dashed curves. For $\calR^\ast<\calR_c$, the trivial disease free solution $z=0$ is the only solution.
\begin{figure}[htb]
\includegraphics[width=0.47\textwidth]{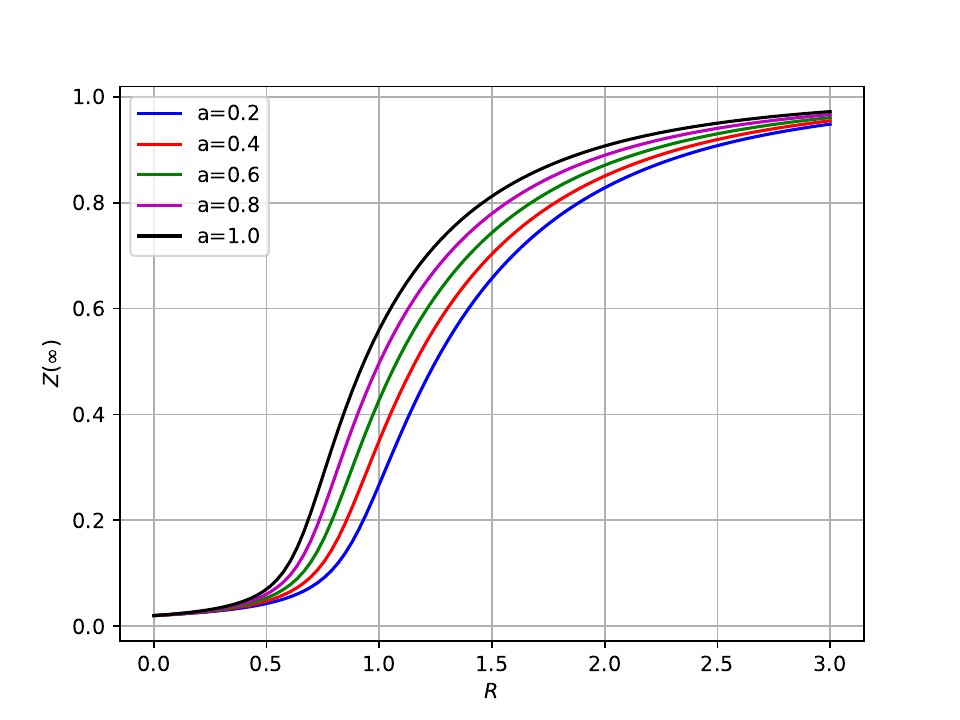}\hfill
\includegraphics[width=0.47\textwidth]{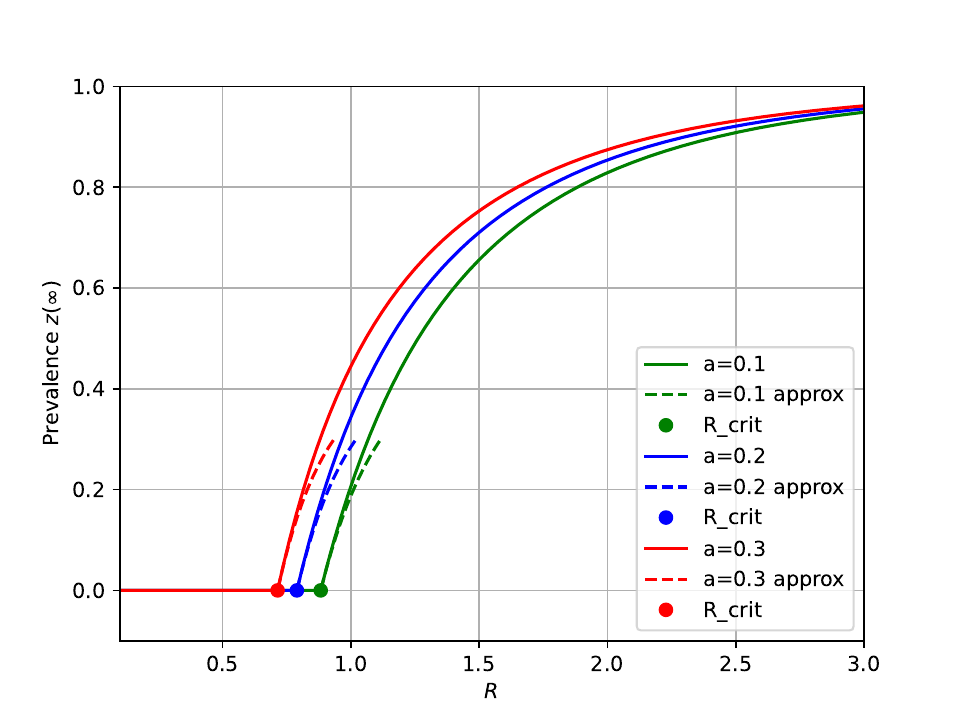}  
\caption{\label{F:preval} Visualization of the prevalence vs.~out--household reproduction number $\calR^\ast$ and for  different in--household attack rate $a$. (Left:) Maximal household size $K=2$. (Right:) Maximal household size $K=3$. Also shown are the asymptotic approximation~\eqref{E:preval3_appx}  and the critical value $\calR_c$.}
\end{figure}

\section{Computing the peak of the infection}
\label{S:Peak}

Let $J= \sum_{k=1}^K k I_k$ denote the total number of infected. Then $J$ satisfies
\begin{align*}
	J' &= \sum_{k=1}^K k I_k' = Y \lsb -\frac{N}{\calR^\ast} + \sum_{k=1}^K k \sum_{j=k}^K j S_j b_{j,k} \rsb
		= Y \lsb -\frac{N}{\calR^\ast} + \sum_{j=1}^K j S_j \sum_{k=1}^j k b_{j,k} \rsb \\
		&= Y \lsb -\frac{N}{\calR^\ast} + \sum_{j=1}^K j S_j E_j \rsb\;.
\end{align*}
For $K=2$ we have to consider the problem
\begin{align*}
	S_1' &= Y\lsb -S_1 + 2S_2 b_{2,1} \rsb \\
	S_2' &= Y\lsb -2 S_2\rsb \\
	J' &= Y \lsb -\frac{N}{\calR^\ast} + S_1 + 2 E_2 S_2 \rsb\;.
\end{align*}
Note, that $E_2 = b_{2,1}+2b_{2,2}$ and $b_{2,1}+b_{2,2}=1$, hence $b_{2,1}=2-E_2$. For sake of shorter notation and easier interpretation, we introduce the scaled compartments $s_1=S_1/N$, $s_2=S_2/N$ and $j=J/N$. Writing $s_1$ as a function of $s_2$, we get
\begin{align*}
	\frac{d s_1}{d s_2} &= E_2-2 + \frac{s_1}{2 s_2}
\intertext{with the solution}
	s_1 &= c \sqrt{s_2} - 2 (2-E_2) s_2\;,
\end{align*}
where $c= 2 (2-E_2) \sqrt{s_{2,0}} + s_{1,0}/\sqrt{s_{2,0}}$.
For $j$ we have the equation
\begin{align*}
	\frac{d j}{d s_2} &= \frac{-1/\calR^\ast+ s_1 +2 E_2 s_2}{-2 s_2}
	=  \frac{1}{2\calR^\ast s_2} - \frac{d s_1}{d s_2} - E_2
\intertext{with the solution given by}
	j &= \frac{1}{2\calR^\ast} \ln \frac{s_2}{s_{2,0}} - s_1(s_2) - E_2 (s_2- s_{2,0})\;.
\end{align*}
The maximum of $j$ is either attained in case of an under--critical epidemic at initial time at the necessary condition $j'=0$ has to hold. Inserting $s_1$ as a function of $s_2$ we arrive at the quadratic equation
\begin{gather*}
	4 (E_2-1) x^2 + cx - \frac{1}{\calR^\ast} = 0
\intertext{ for $x=\sqrt{s_2}$ with the positive root}
	x= \sqrt{s_2} = \frac{c}{8(E_2-1)} 
		\lb \sqrt{1+\frac{16}{\calR^\ast} \frac{E_2-1}{c^2}}-1\rb\;.
\end{gather*}
This solution is only meaningful, if $s_1+2s_2\le 1$, i.e.
\begin{equation}
\label{E:Cond_for_R}
	\frac{1}{\calR^\ast} \le 1 + 2(E_2-1)x^2
\end{equation}
which is an implicit equation for a threshold value of $\calR^\ast$.

To obtain the value of $j$ at the maximum, we plug this root into the above solution for $j$ and arrive at
\begin{align*}
	j_{max} &= 1-\frac{1}{2\calR^\ast}\lb 1+\ln \frac{s_{2,0}}{x^2}\rb - \frac{cx}{2}\;.
\end{align*}
Hence, the peak of infection $j_{max}$ is a function of the out--household reproduction number $\calR^\ast=\beta_k/\gamma_k$, the expected infections in double households $E_2$ and the initial conditions $s_{2,0}$, $s_{1,0}$. In particular, it is independent of the scaling of the recovery periods $\gamma_k$, as can be seen also in the more complex setting presented in Figure~\ref{F:I_Gamma}.

The following Figure~\ref{F:JMax_ana} shows the peak of infection $j_{max}$ versus the out--household reproduction number $\calR$. Solutions are only plotted, if the condition~\eqref{E:Cond_for_R} is satisfied. The left figure shows the situation for $E_2=1.5$ and different initial conditions. The right figure shows the variation with respect to $E_2=1+a$ keeping the initial conditions fixed as $s_{1,0}=0.5$ and $s_{2,0}=0.25$.
\begin{figure}[htb]
\includegraphics[width=0.47\textwidth]{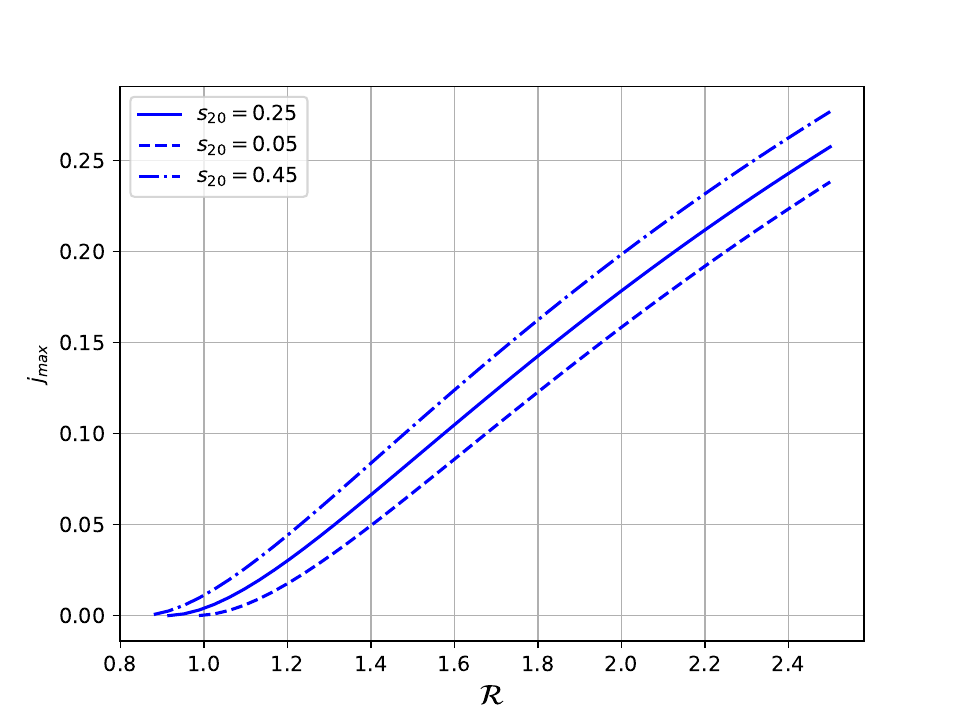}
\hfill
\includegraphics[width=0.47\textwidth]{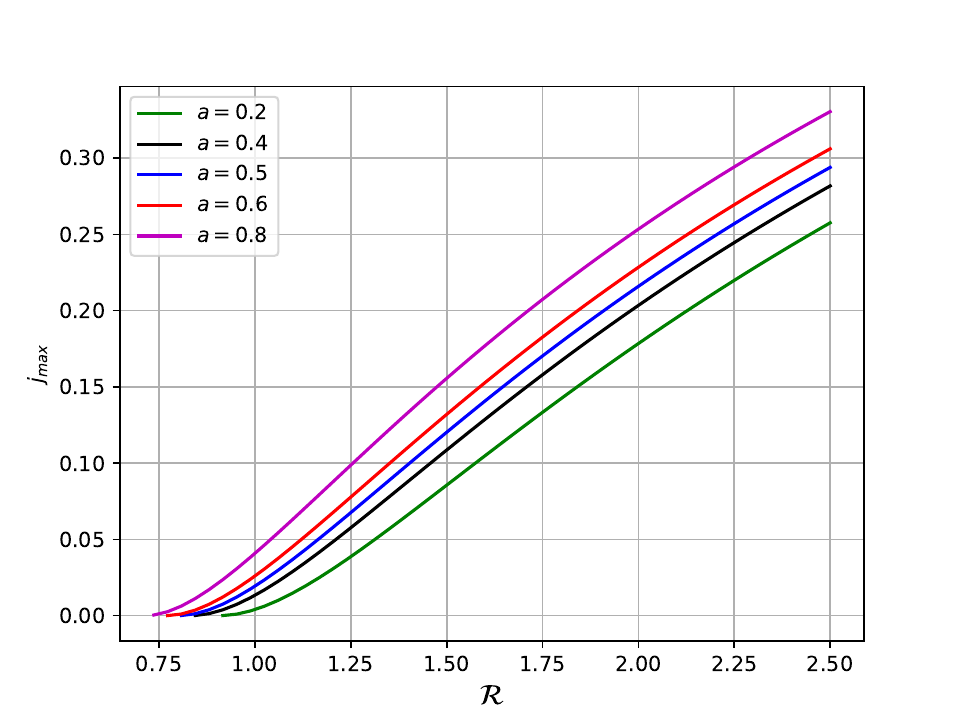}
\caption{\label{F:JMax_ana} Peak height $j_{\max}$ of a single epidemic wave vs.~out--household reproduction number $\calR^\ast$. (Left:) Variation with respect to different initial household size distributions. (Right:) Variation with respect to different in--household attack rates.}
\end{figure}

\section{Simulations and comparison with agent--based model}

The analysis of our household model presented in the previous sections, shows that larger households have a strong effect on the spread of the epidemic. To illustrate this, we simulate a single epidemic wave in populations with close to realistic household distribution. For comparison we choose the household size distributions for Bangladesh (BGD), Germany (GER) and Poland (POL) published by the UN statistics division in 2011, see~\cite{unstats2015demographic}. At the time of writing this paper, the 2011 data sets are the most recent provided by the United Nations for the three countries. Using a sample of almost $1\,000\,000$ individuals, we obtain the  distribution shown in Table~\ref{T:HHsize}.
\begin{table}[h]
{\small
\begin{tabular}{r|r|rrrrrr}
	&& \multicolumn{6}{c}{Number $H_k$ of households of size}\\
	& Total & $k=1$	& $k=2$	& $k=3$	& $k=4$	& $k=5$	& \multicolumn{1}{c}{$k\ge 6$} \\
	\hline
Bangladesh & 999\,995 & 7\,366      & 24\,351     & 44\,022     & 55\,989     & 42\,037     & 53\,960 ($k=7$) \\

Germany	& 999\,999	& 173\,640	& 154\,920	& 67\,846		& 48\,585		& 15\,201		& 7\,106 ($k=6$) \\
Poland	& 999\,996	& 85\,032	& 91\,220	& 71\,221	& 57\,605	& 26\,156	& 22\,523 ($k=7$) \\

\end{tabular}}
\caption{\label{T:HHsize} Distribution of household sizes in Bangladesh, Germany and Poland for the year 2011, see~\cite{unstats2015demographic}. The sample is based on a scaled population of $1\,000\,000$. Differences are due to round--off effects.}
\end{table}

\begin{remark}
Note, that in Table~\ref{T:HHsize} we have redistributed for Poland and Bangladesh the fraction of population living in households of size $6$ or bigger to household of size equal $7$ to match the total population. How one treats the population represented by the tail in the household distribution can make a substantial difference in the simulation outcomes. To visualize that, we show in Figure~\ref{F:Tail} simulations for out--household reproduction number $\calR^\ast=0.9$, in--household attack rate $a=0.2$ and three different versions of how to treat the tail in the Polish household size distribution. The blue curve represents the population distribution given in Table~\ref{T:HHsize}. The green curve uses the extended census data including households up to size $10+$. All households of size $10+$ are treated as households of size \emph{equal} to $10$. The red curve show a simplified treatment of the tail, where all households of size $6+$ are treated as being of size \emph{equal} to $6$.

The simulation redistributing the households of size $6+$ to size $7$ (blue) agrees within the simulation accuracy with the detailed distribution (green). The simplified treatment (red) shows already a significant difference. This difference will be even more pronounced for smaller reproduction numbers $\calR^\ast$ since the smaller households first get subcritical with decreasing $\calR^\ast$. This result clearly visualizes the need for careful treatment of the tail in the household size distribution, since big households have an over--proportional impact on the dynamics. Unfortunately, most publicly available data truncates the household size distribution at size $6$. From a modeling and simulation point of view, there is a need for more detailed data on the household size distribution including information for households of size bigger than $6$.
\begin{figure}[htb]
\centerline{\includegraphics[width=0.67\textwidth]{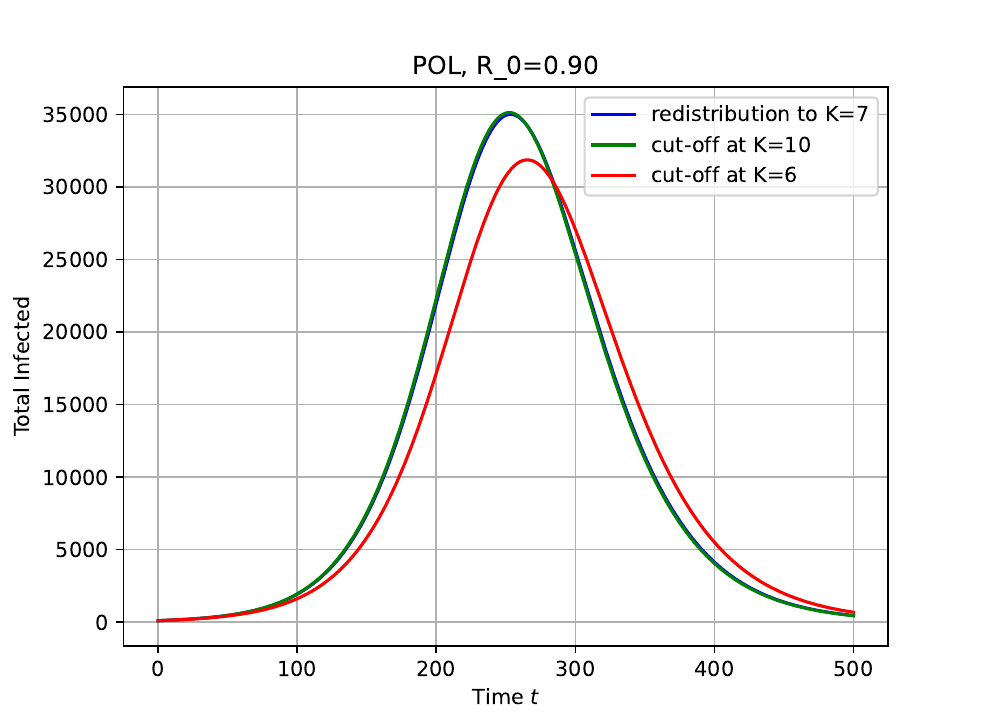}}
\caption{\label{F:Tail} Simulation of one infection wave using different treatment of the tail in the household size distribution. (Blue:) Population distribution given in Table~\ref{T:HHsize} where households of size $6+$ are redistributed to size $7$. (Green:) Polish census data including household sizes up to $10+$. (Red:) Households of size $6+$ treated as being of size equal to $6$.}
\end{figure}
\end{remark}

As initial condition for our simulation we assume $100$ infected single households. The recovery rates inside the households are described using the model of parallel infections, i.e.~$\gamma_k=\gamma_1/(1+1/2+ \cdots + 1/k)$. For the out--household reproduction number $\calR^\ast$ we consider the range $0.6\le \calR^\ast\le 2.3$  and the in--household attack is chosen as $a=0.25$. The ODE--system~\eqref{E:SIR-HH} is solved using a standard RK4(5)--method. 

Figures~\ref{F:Cmp_1} and \ref{F:Cmp_2} show a comparison of the three countries. The prevalence and the maximum number of infected shown in Fig.~\ref{F:Cmp_1} clearly visualize that large households are drivers of the infection. For moderate out--household reproduction number $\calR^\ast\simeq 1$, Bangladesh, with an average household size $\mu_1=4.4$, faces a relative prevalence being approx.~$50\%$ higher than Germany with an average household size of $\mu_1=2.1$ persons. Also the peak number of infected is almost twice as high in Bangladesh compared to Germany for $\calR^\ast\in [1, 1.25]$. Fig.~\ref{F:Cmp_2} shows the simulation of a single infection wave for moderate out--household reproduction number $\calR^\ast=1.1$ and in-household attack rate $a=0.25$. The graph illustrates the faster and more severe progression of the epidemics in case of larger households. The peak of the wave occurs in Bangladesh $60$ days earlier and affects almost double the number of individuals compared to Germany.

\begin{figure}[htb]
\includegraphics[width=0.45\textwidth]{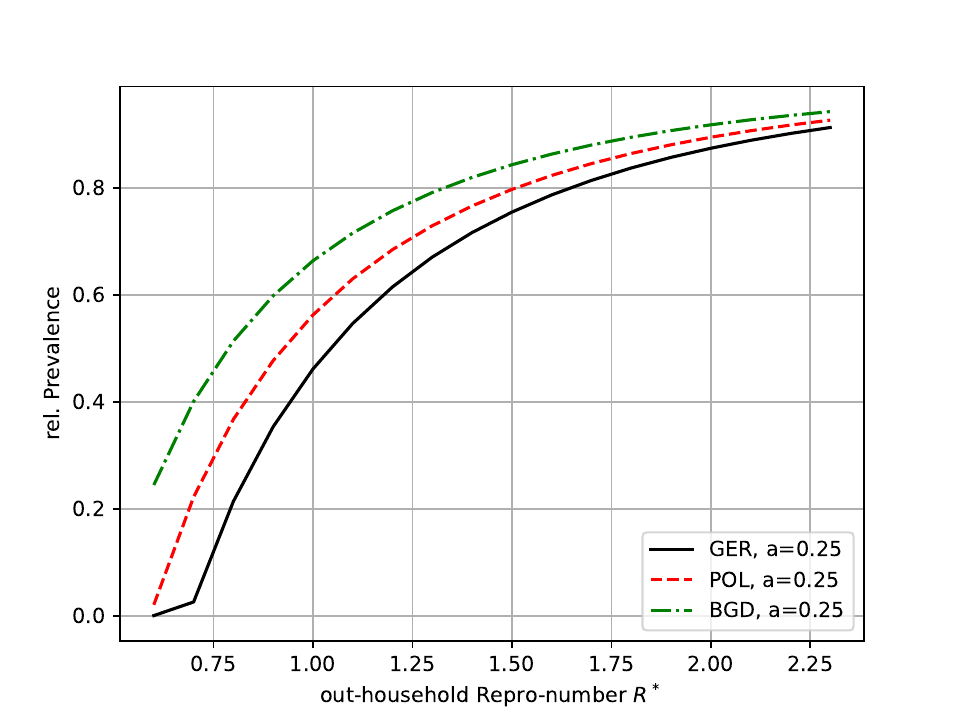} \hfill
\includegraphics[width=0.45\textwidth]{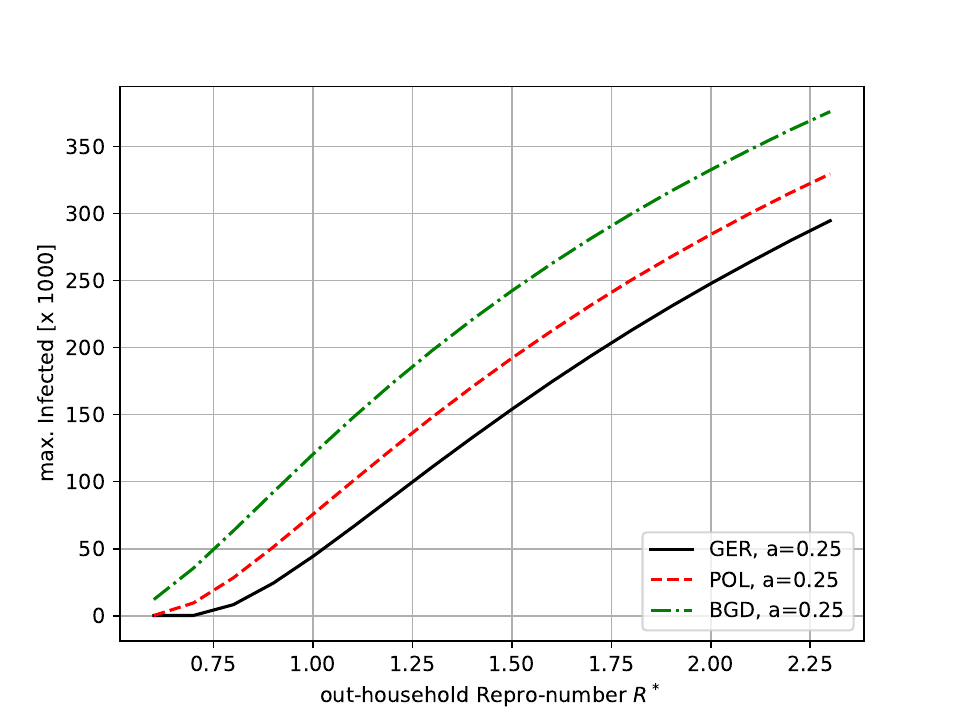}
\caption{\label{F:Cmp_1} Simulation for countries with different household size distributions according to Table~\ref{T:HHsize}. The $x$--axis shows the out--household reproduction number $\calR^\ast$. The in--household attack rate is fixed as $a=0.25$. Left: Relative Prevalence $z$ after $T=2\,000$ days. Right: Absolute height of peak number of infected $J_{\max}$ (in thousands).}
\end{figure}

\begin{figure}[htb]
\centerline{\includegraphics[width=0.67\textwidth]{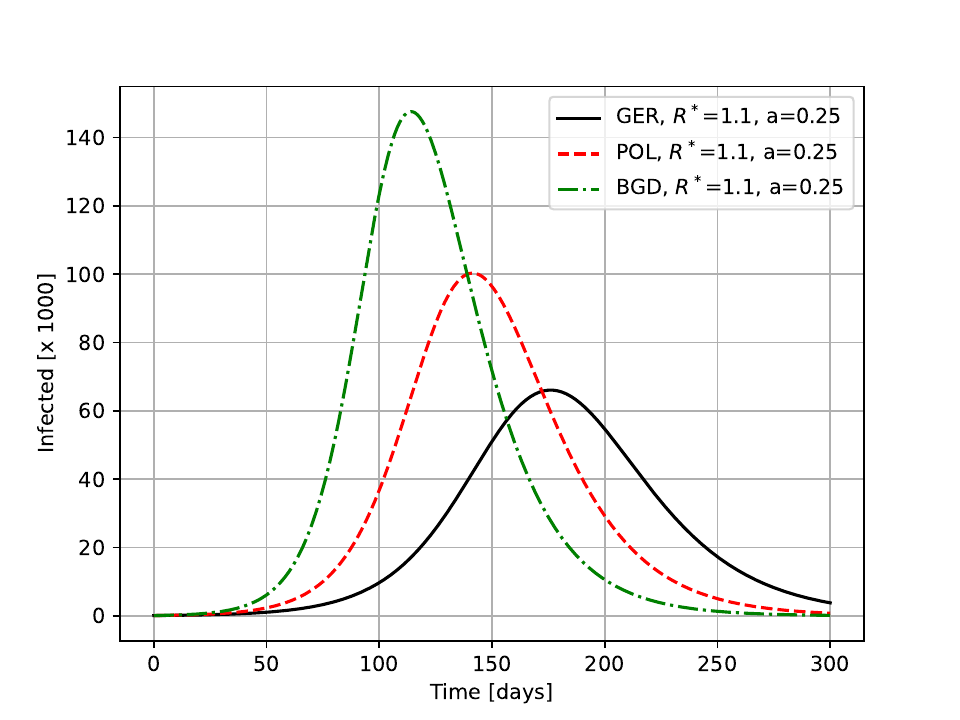}}
\caption{\label{F:Cmp_2} Simulation of one infection wave for countries with different household size distributions according to Table~\ref{T:HHsize}. Out--household reproduction number $\calR^\ast=1.1$ and in--household attack rate $a=0.25$. Plotted is the total number of infected $J$ versus time.}
\end{figure}

To validate our model, we compared it to a stochastic, microscopic agent--based model developed at the  Interdisciplinary Centre for Mathematical and Computational Modelling at the University of Warsaw. Complete details of this model are given in~\cite{niedzielewski2022overview}. In this model agents have certain states (susceptible, infected, recovered) and infection events occur in certain context, i.e.~inside a single household or between different households. Since the ODE--model~\eqref{E:SIR-HH} is a variant of an SIR--model, the agent--based model also just uses the SIR--states and ignores all other states. The agent--based model uses for each infected individual a recovery time that is sampled from an exponential distribution with mean $10$ days. Based on the household distribution for Poland, Figure~\ref{F:POL_prev} provides a comparison of the computed prevalence vs.~the out--household reproduction number for our model~\ref{E:SIR-HH} and the agent--based model (blue squares). The solid lines show the results of the ODE--model for different in--household attack rates. The relative prevalence plotted in the figure is defined as the total number of recovered individuals for time $t\to \infty$; here we use $T=2\,000$ days for practical reasons. 

\begin{figure}[htb]
\centerline{\includegraphics[width=0.67\textwidth]{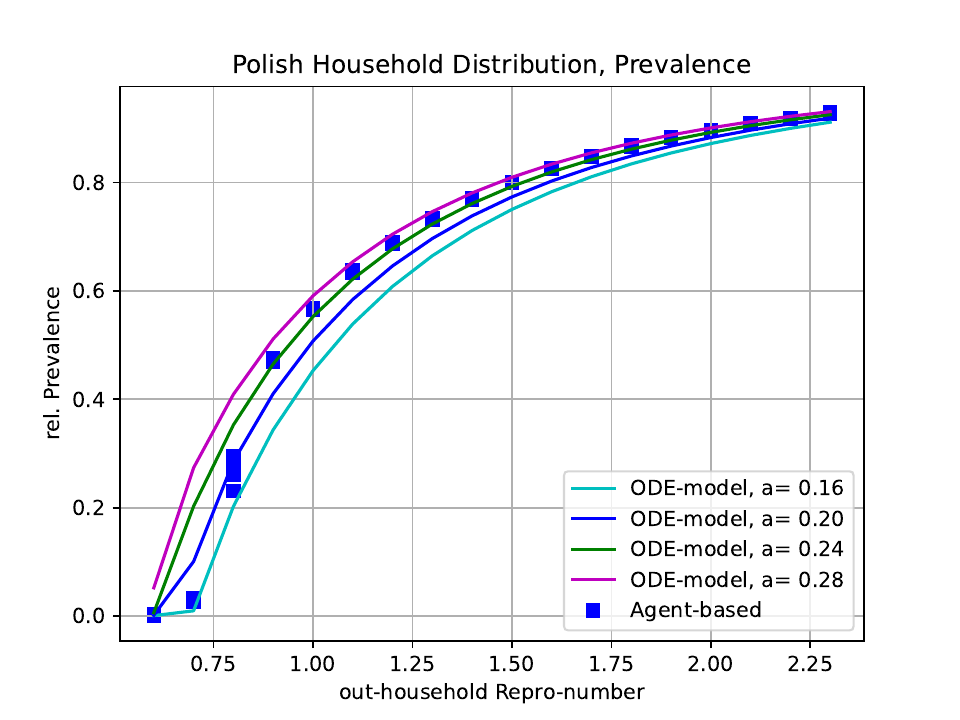}}
\caption{\label{F:POL_prev} Simulation of one infection wave, plotted is the relative prevalence $z$ versus the out--household reproduction number $\calR^\ast$. Initially $100$ infected single households in a population according to Table~\ref{T:HHsize}. The solid lines show the results of the ODE model for different in--household attack rates. The results of the agent--based model are shown by the blue squares.}
\end{figure}

Figure~\ref{F:POL_peak} shows the results for the peak of the infection wave. The graph on the left shows the peak of the incidences, i.e.~the maximum number of daily new infections and the graph on the right shows the time, when this peak occurs. 
\begin{figure}[htb]
\includegraphics[width=0.45\textwidth]{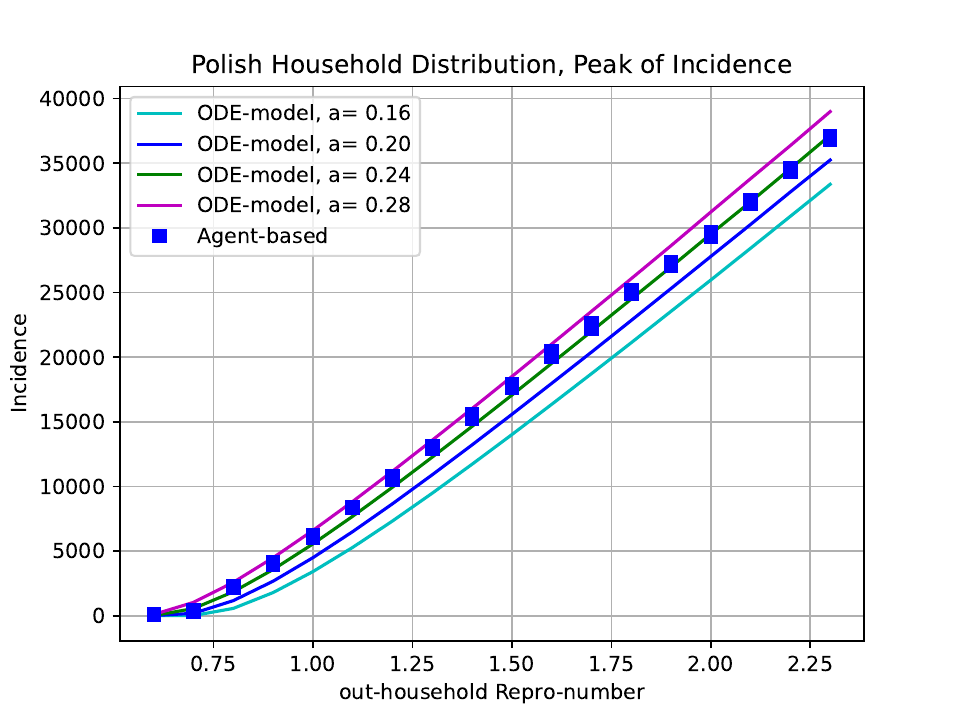} \hfill
\includegraphics[width=0.45\textwidth]{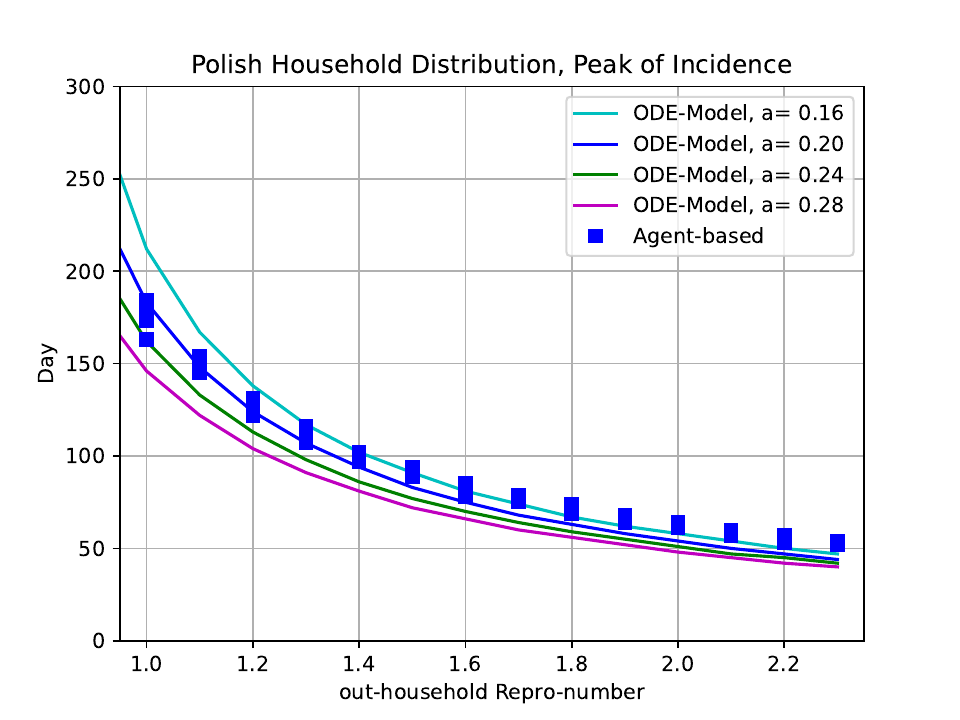}
\caption{\label{F:POL_peak} Simulation of one infection wave, initially $100$ infected single households in a population according to Table~\ref{T:HHsize}. Left: Peak number of daily new infections versus the out--household reproduction number $\calR^\ast$. Right: Time, when the peak occurs vs.~$\calR^\ast$. The solid lines show the results of the ODE model for different in--household attack rates. The results of the agent--based model are shown by the blue squares.}
\end{figure}

For all three presented criteria: prevalence, peak height and peak time, our ODE--household model~\eqref{E:SIR-HH} matches quite well with the agent--based simulation. Best agreement can be seen for in--household attack rate $a$ in the range between $0.16$ and $0.20$.

\section{Effect of household quarantine}

In this section we will extend our basic model~\eqref{E:SIR-HH} to households put under quarantine after an infection is detected. Therefore, we introduce the additional compartments $Q_k$ of quarantined households of size $k$. The extended household model reads as
\begin{subequations}
\label{E:SIQR-HH}
\begin{align}
	S_k' &= Y \lsb - k S_k + \sum_{j=k+1}^K\!\! j S_j \cdot b_{j,j-k} \rsb \\
	I_k' &= - \gamma_k I_k - q_k I_k + Y \sum_{j=k}^K j S_j \cdot  b_{j,k} \\
	Q_k' &= q_k I_k - \gamma_k Q_k \\
	R_k' &= \gamma_k I_k + \gamma_k Q_k
\end{align}
\end{subequations}
Here $q_k>0$ denotes the detection and quarantine rate for a household of size $k$.

Let $q>0$ be the probability of infected individual to get detected. The probability that a household consisting of $k$ infected gets detected equals $d_k(q) = 1-(1-q)^k$. Given recovery and detection rates $\gamma_k$ and $q_k$ for a household of size $k$, the probability, that the household gets detected before recovery equals $q_k/(\gamma_k+q_k)$, which equals $d_k(q)$. Therefore
\begin{equation*}
	\frac{q_k}{\gamma_k} = \frac{d_k(q)}{1-d_k(q)} \sim kq
\end{equation*}
for $q\ll 1$.

To compute the basic reproduction number, we repeat the computations for the next generation matrix. The vector $\xi$ of infected compartments remains unaltered and the new quarantine compartments $Q_1,\dots, Q_K$ are included in the vector $\chi$. The gain term $\calF_k$ for the infected compartments remains unaltered as well, but the loss term reads as $\calV_k= (\gamma_k+q_k)\xi_k$. Accordingly, its Jacobian modifies to $V_{kk} = \gamma_k+ q_k$ and finally the next generation matrix $G$ reads as
\begin{align*}
	G &= \frac{\beta}{N} 
		\lb \frac{j}{\gamma_j+q_j}\sum_{m=i}^K m H_m\, b_{m,i} \rb_{i,j} \\
		&= \calR^\ast 
		\lpm \sum_{m=1}^K \tfrac{m H_m}{N}\, b_{m,1} \\[.5ex] 
			\sum_{m=2}^K \tfrac{m H_m}{N} \, b_{m,2} \\\vdots \\
			\tfrac{K H_K}{N} \,  b_{KK} \rpm 
	\cdot \lpm \frac{1}{1+q_k/\gamma_k},\frac{2}{1+q_2/\gamma_2}, \dots, \frac{K}{1+q_K/\gamma_K}\rpm
\end{align*}
Its non--zero eigenvalue equals 
\begin{equation*}
	\calR_q := \calR^\ast \sum_{k=1}^K \frac{k}{1+q_k/\gamma_k}\sum_{m=k}^K \frac{m H_m}{N}b_{m,k}
	= \calR^\ast \sum_{m=1}^K \frac{m H_m}{N} \sum_{k=1}^m \frac{kb_{m,k}}{1+q_k/\gamma_k}\;.
\end{equation*}
Due to denominator $1+q_k/\gamma_k$ we cannot interpret the last sum as the expected infected cases in a household of size $m$. In the asymptotic case of small quarantine rates $q_k\ll 1$, we can use the series expansion $(1+q_k/\gamma_k)^{-1} \sim 1- q_k/\gamma_k\sim 1-kq$ and get
\begin{align*}
	\calR_q \sim \calR^\ast \sum_{m=1}^K \frac{m H_m}{N} \lb E_m - q \sum_{k=1}^m k^2 b_{m,k}\rb\;.
\end{align*}
So, the second moment of the in--household infection distribution comes into play.

\begin{remark}
If the expectation $E_m$ and the the second moment $\sum_{k=1}^m k^2 b_{m,k}$ of the in--household infections scale linear with the household size $m$, then the above reproduction number depends on also on the third moment of the household size distribution. This leads to a straight--forward extension of the result~\eqref{E:tildeR}.

In case of a binomial distribution of the in--household infections, it holds that $E_m=a(m-1)+1$ and $\sum_{k=1}^m k^2 b_{m,k}=a(m-1)\lb 3+a(m-2)\rb+1$. For the reproduction number we arrive at the expression
\begin{equation*}
	\calR_q \sim \calR^\ast \lsb 1-\frac{q}{\gamma}
		+ a \lb 1- \frac{3q}{\gamma}(1-a)\rb \lb \frac{\mu_2}{\mu_1}-1\rb 
		- \frac{q}{\gamma}a^2 \lb \frac{\mu_3}{\mu_1} - 1 \rb \rsb\;,
\end{equation*}
where $\mu_3 := \sum_{k=1}^K k^3 h_k$ denotes the third moment of the household size distribution.
\end{remark}

The effectiveness of quarantine measures can be assessed by relating the reduction in prevalence to the quarantined individuals. For a given quarantine probability $q>0$, let $Z(q)$ denote the prevalence observed under these quarantine measures and let $Q(q)$ denote the total number of individuals quarantined. We define the effectiveness of the quarantine measures as
\begin{equation*}
    \eta_q := \frac{\abs{Z(q)-Z(0)}}{Q(q)}\;.
\end{equation*}
The numerator equals to the reduction in prevalence and hence $\eta$ can be interpreted as infected cases saved per quarantined case.

In Figure~\ref{F:quarantine} we compare the effectiveness in the setting of household distributions resembling Germany (predominantly small households) and Bangladesh (large households are dominant) for out--household reproduction numbers in the range between $0.9$ and $1.2$. The results indicate, that quarantine measures are less effective in societies with larger households since infection chains inside households are not affected by the quarantine. On the other hand, quarantine measures seem most effective in presence of small households and for low  reproduction numbers.

\begin{figure}[htb]
\centerline{\includegraphics[width=0.67\textwidth]{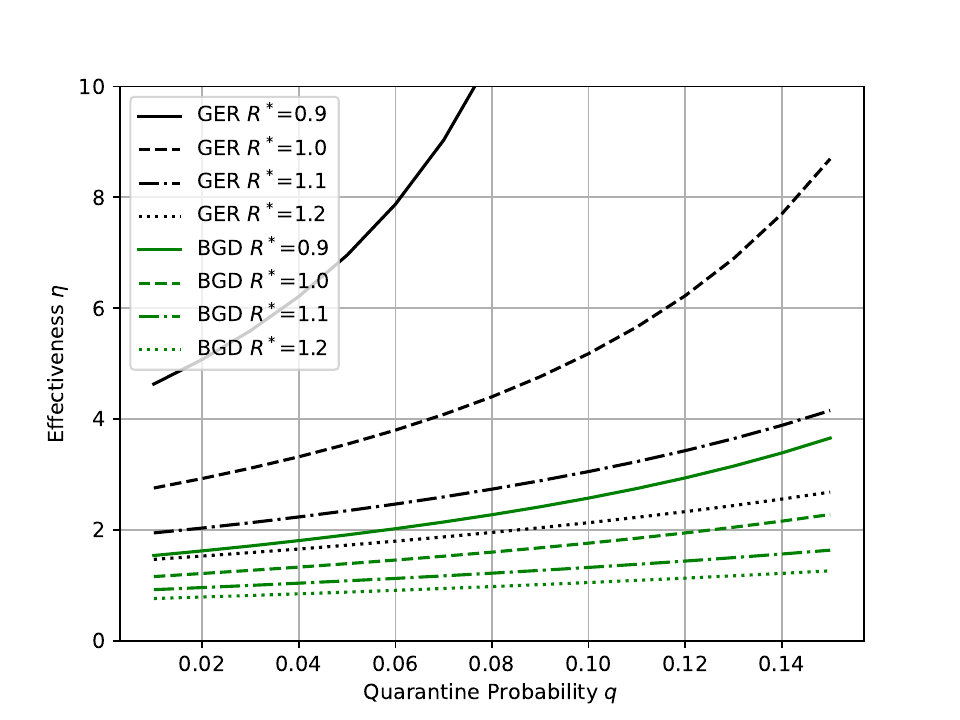}}
\caption{\label{F:quarantine} Effectiveness $\eta$ of quarantine measures for the different household distributions in Germany (GER) and Bangladesh (BGD), see Table~\ref{T:HHsize}.}
\end{figure}

\section{Conclusion and Outlook}
The findings of the compartmental household model~\eqref{E:SIR-HH} are in rather good agreement with microscopic agent--based simulations. Despite its several simplifying assumptions, the prevalence and the peak of the infection wave are reproduced quite well. 

However, field data suggests that the in--household attack rate $a$ is significantly higher for $2$--person households than for larger households, see~\cite{house2022inferring, madewell2020household}. Analyzing the data given in~\cite{house2022inferring} the attack rate in two--person households is by a factor of $1.25$--$1.5$ larger than in households with three or more members. According to~\cite{madewell2020household} this may be caused by spouse relationship to the index case reflecting intimacy, sleeping in the same room and hence longer or more direct exposure to the index case. An extension of the model~\eqref{E:bjk} to attack rates $a_k$ depending on the household size $k$ is straightforward. Carrying over the analytical computation of the reproduction number~\eqref{E:tildeR} would lead to a variant of~\eqref{E:tildeR} including weighted moments of the household size distribution.

When considering quarantine for entire households, the model shows some interesting outcomes. Defining the ratio of reduction in prevalence and number of quarantined persons as the effectivity of a quarantine, the model suggests that quarantine is more effective in populations where small households dominate. Whether this finding is supported by field data is a topic for future literature research.

Vaccinations are not yet included in the model. From the point of view of public health, the question arises whether the limited resources of vaccines should be distributed with preference to larger households. Given the catalytic effect of large households on the disease dynamics, one could suggest to vaccinate large households first. However, a detailed analysis of this question will be subject of follow--up research.

When extending our ODE SIR--model to an SIRS--model including possible reinfection a recombination of the recovered sub--households is required. Again, this extension will be left for future work. A relaxation of the distribution of the recovered sub--household distribution to the overall household distribution seems a tractable approach to tackle this issue.

\appendixpage
\appendix

\section{Microscopic Theory}

In this appendix we would like to contrast the results on the basic reproduction number, see~\eqref{E:tildeR} and on the prevalence, see~\eqref{E:preval} and~\eqref{E:z} with results obtained by asymptotic microscopic theory. As expected, one gets the same criterion for the epidemic threshold as a function of $\calR^\ast$ and the household structure. Additionally it turns out that the computation of the asymptotic prevalence $z$ is more transparent in this setting and can be carried out in great generality once proper assumptions about the in--household infections have been made.

In a microscopic model, we can represent SIR infection dynamics by an Erdös--Rényi random graph $G(N,p)$ \cite{erdoesrenyi1960}. The $N$ nodes of the graph represent the individuals under consideration and the random edges represent infectious contacts between two individuals. Assuming homogeneous mixing and random contact structures, edges are created independently with probability $p$ between each pair of the $N$ nodes. The $G(N,p)$--random graph model, where $p=\frac{\calR^\ast}{N}$ corresponds to an SIR--model with reproduction number $\calR^\ast$ in the limit of large $N$. To capture the timeline of epidemic dynamics one would still need independent exponential Exp$(\gamma)$--distributed ''travel times'' assigned to the edges corresponding to the recovery rate $\gamma$ in the ODE--SIR model.

In the following we concentrate only on results depending on the structural properties ignoring all aspects of timing. As in the simple SIR case we will obtain that the epidemic phase transition in a microscopic model including households depends only on the out-household reproduction number, the in--household attack rate and the household size distribution. The asymptotic prevalence, that is the fraction of population getting infected eventually, does only depend on the structural properties of the underlying random graph model and is given essentially by the size of the largest connected component --- the so--called giant component in case it
exists.

\subsection{Household graphs and phase transition}
\label{S:App1}

As in the ODE model we assume a population of size $N$ partitioned into households of size $k=1,\dots, K$. Analogously to Section~\ref{S:Model}, let $\lcb h_k\rcb$ denote the household size distribution, and $h_{k}$ equals to the fraction of households of size $k$. Furthermore, let $\mu_1$ and $\mu_2$ denote the first and second moment of the household distribution. 

We model the contacts between individuals as an Erdös--Rényi random graph. A directed edge from individual $i$ to individual $j$ represents an infectious contact, where $i$ infects $j$. The probability that such an edge exists is given by
\begin{equation}
	\Pr \lb i\leadsto j\rb =\frac{\calR^\ast}{N}\;.
\end{equation}
Hence, in the limit $N\to \infty$, the reproduction number equals $\calR^\ast$. 

Analogous to the ODE model, households are composed of individuals and hence directed graph between individuals induce also a directed graph between the households. Let $x=(x_1,x_2)$ denote a household of size $x_2$ with $x_1$ infected members. Then, the probability on an infectious contact between household $x$ and another household $y$ equals 
\begin{equation}
	\Pr \lb x\leadsto y\rb =\frac{\calR^\ast}{N} x_{1}\cdot y_{2}
		=\frac{\calR^\ast}{\mu_{1}\cdot H}x_{1}\cdot y_{2}\;,
\end{equation}
where $H$ denotes the total number of households.

As in the ODE model, we denote by $b_{k,l}$ the probability that a primary infection brought into a household of size $k$ produces additional $l-1$ infected household members via the chain of in--household contacts. Since in--household infections and out--household infections are assumed to be independent we have an a priory distribution of household types given by the measure $\nu \lb x\rb =h_{x_{2}}\cdot b_{x_{2},x_{1}}$.

We can now interpret our random network of households as a special case of
heterogenous random graphs as described in Bollobás~\cite{bollobas2007randgraph} for undirected graphs and in Cao~\cite{cao2020digraphs} for directed graphs. We recall the basic construction of Bollobás--Janson--Riordan (BJR) graphs. In a BJR--graph each of the $N$ nodes get assigned a feature variable $x$ from a measurable feature space $\lb\calS,\nu \rb$ where $\nu$ is the asymptotic distribution of features as $N\to \infty$. Edges are defined via a positive kernel function $\varkappa: \calS\times \calS\to \R^+$ such that the probability of a directed edge from
a node of type $x$ to node of type $y$ equals 
\begin{equation}
	\Pr \lcb x\leadsto y\rcb =\min \lb 1, \frac{\varkappa(x,y)}{N}\rb\;.
\end{equation}
In our setting the connections between households are modelled by the kernel function $\varkappa(x,y) =\frac{\calR^\ast}{E}x_{1}\cdot y_{2}$. All conditions required for the kernel function in~\cite{bollobas2007randgraph} are trivially satisfied, since we only consider  a finite number of household types and the kernel is irreducible.

In the framework of BJR--graphs, the transfer operator 
\begin{equation}
	(Tf) (x) := \int_\calS \varkappa(x,y)\, f(y)\, d\nu
\end{equation}
acting on the Hilbert space $L_{2}\lb S,d\nu\rb$ plays a crucial role. The phase
transition is given by the criterion $\rho(T)=1$ on the spectral radius of the operator. Note that the spectral norm $\rho(T)$ can be interpreted as the mean basic reproduction number $\calR_0$, see~\cite{bollobas2007randgraph}.

In our setting the kernel function is of product type $\varkappa(x,y) =\varphi(x)\cdot \cdot \psi(y)$, where $\varphi(x)=\frac{\calR^\ast}{\mu_1}x_{1}$ and $\psi(y)=y_{2}$. In this rank one situation the spectral radius of $T$ is
given by 
\begin{align*}
	\rho(T) &= \int_\calS \phi(x)\cdot \psi(x)\, d\nu(x) \\
	&= \frac{\calR^\ast}{\mu_1} \sum_{x} x_1\cdot x_2
				\cdot h_{x_2}\cdot b_{x_2,x_1}\;.
\intertext{Let $a_m=\frac{1}{m-1}\lb \sum_{k=1}^m k\, b_{k,m}-1\rb$ denote the expected number of in--household infections if a household of size $m$ gets infected from the outside. Then}
	\rho(T) &= \frac{\calR^\ast}{\mu_1} \sum_{x_2} \lb 1+a_{x_2}(x_2-1)\rb 
				\cdot x_2\cdot h_{x_2} \\
		&= \frac{\calR^\ast}{\mu_1} \lsb \mu_1 +  \sum_{x_2} a_{x_2}(x_2-1)
				\cdot x_2\cdot h_{x_2}\rsb\;.
\end{align*}
Note that $a_{m}$ is called in epidemiological literature the secondary in--household attack rate in a household of given size. In the special case when $a_{m}=a$ is independent of the actual household size, we recover the analogue to Eqn.~\eqref{E:tildeR}
\begin{equation}
	\label{E:normT}
	\rho(T) = \calR^\ast \lsb 1+a\lb \frac{\mu_2}{\mu_1}-1\rb \rsb\;. 
\end{equation}

\subsection{The prevalence}
\label{S:App2}

Let $J(0)$ denote the total number of infected individuals at initial time and let $Z(\infty)$ denote the total number of recovered individuals at final time. We define the asymptotic prevalence $z=\frac{Z(\infty)}{N}$ defined as the fraction of recovered individuals assuming that we start with a small (compared to $N$) but still large number $J(0)$ of initially infected individuals. Contrary to the phase transition threshold which only depends on the in--household attack rates $a_m$, the moments of the household size distribution and the out--household reproduction number $\calR^\ast$, the situation for the prevalence $z$ is more involved.

As in Section~\ref{S:Model}, let $H_{k}$ denote the number of households of size $k$ and $H=\sum_k H_{K}$ equals to the total number of households. The moment $\mu_1$ equals to the average household size and $N=\sum_{k} kH_k=H\cdot \mu_1$.

We first compute the probability $p_O$ that an individual had an infectious out--household contact. Note, that having an infectious out--household contact does not necessarily imply that the individual got infected via an out--household contact since it could habe been infected earlier via an in--household contact. 

The prevalence $Z(\infty)=Nz$ equals to the total number of individuals that ever got infected. In the random graph setting, the probability of an infectious out--household contact equals $\frac{\calR^\ast}{N}$ and hence the probability to escape from an infectious contact with all infected individuals from outside the household equals $\lb 1-\frac{\calR^\ast}{N}\rb^{zN}$. Since $p_O$ equals to the probability \emph{not} to escape from out--household infection, we get 
\begin{align}
	p_O &= 1 - \lb 1-\frac{\calR^\ast}{N}\rb^{zN} \notag
\intertext{and in the limit $N\to \infty$}
	p_O &= 1- \exp \lb-\calR^\ast\cdot z\rb\;. \label{E:pO}
\end{align}
Let $v_k=v_k(p_O)$ denote the expected number of infected individuals in a household of size $k$ given the probability $p_O$, that a household member gets infected from the outside. Obviously, we have 
\begin{align*}
	\sum_k H_k v_k &= z\cdot N
\intertext{and hence}
	\sum_k h_k v_k &= z\cdot \mu_1\;.
\end{align*}
The number of times a household of size $k$ is exposed to an infection from outside is $\text{Bin}\lb p_O,k\rb$--distributed. Therefore, we can express $v_k$ as
\begin{equation}
	\label{E:vk_of_pO}
	v_k(p_O) = \sum_{l=1}^k d_{l,k} \binom{k}{l}\, p_O^l\, \lb 1-p_O\rb^{k-l}\;,
\end{equation}
where $d_{l,k}$ is the expected number of infected individuals in a household of size $k$ if $l\geq 1$ household members got exposed to an infection from outside. To compute $d_{l,k}$ we need more specific assumptions about the in-household infection process. Let $P_l\lb m,k\rb$ be the probability that $m$ individuals in a household of size $k$ get infected if $l$ randomly chosen household members got exposed to an infection from outside. Then
\begin{equation}
	\label{E:dlk}
	d_{l,k}= \sum_{m=l}^k m P_l(m,k) \;.
\end{equation}
Let $P_1(m,k)=b_{k,m}$ be given. For the probabilities $P_l(m,k)$ we have the following recursion in with respect to $l$.
\begin{equation}
	\label{E:P_recursive}
	P_{l+1}(m,k) =\sum_{m'=l}^{m-1} 
			P_l(m',k)\cdot P_1(m-m',k-m')\cdot \frac{k-m'}{k-l} 
			+ P_l(m,k)\cdot\lb 1-\frac{k-m}{k-l}\rb\;.
\end{equation}
The terms under the summation describe the events that after $l$ exposures $m'<m$ household members got infected. The next, $(l+1)$th exposure targeted a yet non--infected individual with probability $\frac{k-m'}{k-l}$, and after this $l+1$ exposure additional $m-m'$ new infections
appear with probability $P_{1}(m-m',k-m')$. The last term on the right hand side is just the probability that after $l$ exposures already $m$ individuals got infected and the $(l+1)$th exposure did not generate any new infection. Note that we can rewrite the the terms $P_{1}(a,b)$ appearing in the summation by expression explicitly depending on the household size $k$.

Having computed $v_k(p_O)$ via~\eqref{E:vk_of_pO}, \eqref{E:dlk} and \eqref{E:P_recursive} as a function of $b_{k,l}$ and $p_O(z)=1-e^{-\calR^\ast\cdot z}$ depending on the prevalence $z$, one can solve the fixed point equation 
\begin{equation*}
	\sum_{k\ge 1} h_k v_k(p_O(z)) = \mu_1 \cdot z\;.
\end{equation*}
In case of $\rho(T)>1$, there exists a a unique, nontrivial solution $0<z\le 1$. This fixed point equation is a generalization of the results~\eqref{E:preval} and~\eqref{E:z} obtained from the ODE--SIR model. Structurally, the fixed point problem is a polynomial in $e^{-\calR^\ast z}$.

We close this section with the computation of $b_{k,m}$ for an important
special case. Assume every infected household member infects non--infected
household members with equal probability $p=p(k)$ just depending on the total household size. The total number of infections after an initial infection is brought into the household from the outside, is the result of a possible cascade of consecutive in--household infection events. Let $b_{k}(p)=b_{k,k}$ denote the probability that in a household of size $k$ \emph{all} household members get infected. We have the following recursion in $k$ : 
\begin{align*}
	b_k(p) &= 1-\sum_{m=1}^{k-1} 
		\binom{k-1}{m-1}\, b_m(p)\cdot \lb 1-p\rb^{m(k-m)}\;,
\intertext{where the summation is over all events that less than $k$ household members got infected. Finally we obtain}    
	b_{k,m} &= \binom{k-1}{m-1} \cdot b_m(p) \lb 1-p\rb^{m(k-m)}\;.
\end{align*}

\section*{Acknowledgments}
This collaboration between the groups in Koblenz and in Poland was funded by the DAAD--NAWA joint project ''MultiScale Modelling and Simulation for Epidemics''--MSS4E, DAAD project number: 57602790, NAWA grant number: PPN/BDE/2021/1/00019/U/DRAFT/00001. 

The ICM model was developed as part of the ''ICM Epidemiological Model development'' project, funded by the Ministry of Science and Higher Education of Poland with grants 51/WFSN/2020T and 28/WFSN/2021 to the University of Warsaw.\\
We would like to thank the ICM Epidemiological Model team members for development of a new software engine for agent-based simulations (pDyn2) that was used in generation of part of the results of this work.

TK thanks Wroclaw University of Science and Technology for providing the necessary infrastructure and scientific environment for several meetings of the authors in Wroclaw within the NAWA project. 

TK and PD thank the RESPINOW project (BMBF, project ID 031L0298D) for financing several visits of TK to the MPI Göttingen and one visit of PD to Wroclaw University of Science and Technology.

VP was supported by the infoXpand project within MONID (BMBF, project ID 031L0300A), and by the German Research Foundation (DFG), SFB 1528--Cognition of Interaction.

\bibliography{references}
\bibliographystyle{siam}

\end{document}